\documentclass[opre,nonblindrev]{informs3}

\DoubleSpacedXI 

\usepackage{endnotes}
\usepackage{algorithm}
\usepackage{algpseudocode}
\let\footnote=\endnote
\let\enotesize=\normalsize
\def\notesname{Endnotes}%
\def\makeenmark{\hbox to1.275em{\theenmark.\enskip\hss}}
\def\enoteformat{\rightskip0pt\leftskip0pt\parindent=1.275em
  \leavevmode\llap{\makeenmark}}

\usepackage{natbib}
 \bibpunct[, ]{(}{)}{,}{a}{}{,}%
 \def\bibfont{\small}%

 \definecolor{myblue}{RGB}{0, 20, 114}
\usepackage[hidelinks,colorlinks=true,urlcolor=myblue,linkcolor=myblue,citecolor=myblue]{hyperref}
\def\EMAIL#1{\href{mailto:#1}{#1}}

\def\theARTICLETOP{}
\def\theLRHFirstLine{}
\def\theLRHSecondLine{}
\def\theRRHFirstLine{}
\def\theRRHSecondLine{}
\usepackage{tikz}
\usepackage{enumitem}
\usepackage{subcaption}

\TheoremsNumberedThrough     
\ECRepeatTheorems

\EquationsNumberedThrough    

\begin{document}


\RUNAUTHOR{Hu, Hu and Zhou}

\RUNTITLE{Quantum Monte Carlo Tree Search with Fixed
Confidence}

\TITLE{Quantum Monte Carlo Tree Search with Fixed
Confidence}

\ARTICLEAUTHORS{%
\AUTHOR{Mingjie Hu}
\AFF{School of Management, Fudan University \\
H. Milton Stewart School of Industrial and Systems Engineering, Georgia Institute of Technology\footnotemark[1]\\
\EMAIL{23110690009@m.fudan.edu.cn}}
\AUTHOR{Jian-Qiang Hu}
\AFF{School of Management, Fudan University, \EMAIL{hujq@fudan.edu.cn }}\AUTHOR{Enlu Zhou}
\AFF{H. Milton Stewart School of Industrial and Systems Engineering, Georgia Institute of Technology, \EMAIL{enlu.zhou@isye.gatech.edu}}
} 
\footnotetext[1]{This work was done during Mingjie Hu's visit at Georgia Institute of Technology.}

\ABSTRACT{%
Recent advances in quantum computing are opening new opportunities for computationally intensive decision problems. This paper studies how quantum computing can improve Monte Carlo tree search (MCTS) in the fixed-confidence setting, where the goal is to identify a near-optimal move in a given game tree with high probability while minimizing query complexity. We first formulate MCTS under a quantum oracle model. We then
develop a quantum MCTS algorithm (QMCTS) that combines threshold-based elimination on the search tree with quantum Monte Carlo estimation to reduce the cost of evaluating stochastic leaf values. We establish an instance-dependent lower bound on the query complexity and derive a corresponding upper bound for QMCTS. The lower-bound analysis introduces a new quantum phase-testing result that may also be useful beyond MCTS. We validate the theoretical results through simulation experiments and further demonstrate the feasibility of QMCTS on real quantum hardware.
}%


\maketitle

%


\section{Introduction}
\label{sec: intro}
Monte Carlo three search (MCTS) is a search framework aiming to find the optimal decision based on the search tree built by random sampling of the decision space \citep{shah2022nonasymptotic}. It has achieved remarkable success in sequential decision making with a tree representation, such as the game of Go \citep{silver2017mastering_}, chess, and shogi \citep{silver2017mastering}. The central statistical task is to compare candidate actions using noise simulation outcomes and take the optimal action. The effectiveness of MCTS, however, can be limited by the number of simulation required to identify the best move for a given node on the tree.  

Quantum computing has advanced rapidly and has shown promise in areas such as combinatorial optimization \citep{abbas2024challenges}, simulation optimization \citep{hu2025quantum,hu2026quantum}, and machine learning \citep{biamonte2017quantum}. For certain computational tasks, quantum algorithms can achieve substantial speedups over their classical counterparts. This naturally raises the following question: 
    \textit{What can quantum computing offer to MCTS?}
This question is relevant to a broad range of learning and decision-making problems in which MCTS serves as a key computational tool.

{\color{black}To answer this question, we adopt the fixed-confidence MCTS formulation of \citet{kaufmann2017monte} and consider a two-player zero-sum game in which the possible sequences of play are represented by a maximin game tree.} MAX nodes correspond to decisions made by the first player, while MIN nodes correspond to decisions made by the second player. The terminal payoffs are stochastic, and the objective is to identify a root move whose value is within $\epsilon>0$ of the optimal root value with probability at least $1-\delta$. 

However, designing a quantum MCTS algorithm raises several fundamental technical challenges:
\begin{itemize}
    \item \textbf{Quantum Information Extraction.} The advantage of quantum computing relies on processing information encoded in quantum superposition. However, extracting classical information (such as the mean estimate) requires \emph{measurement}, which collapses the quantum state and prevents its further use. This constraints makes many classical sequential algorithms unsuitable in the quantum setting and requires new algorithmic designs.
    \item \textbf{Tree-Structured Exploration.} The recursive MAX-MIN structure couples the leaf values, and the effect of each leaf depends on its position in the tree. This structural dependence makes it difficult to characterize the intrinsic difficulty of a problem instance and to obtain matching upper and lower bounds on the query complexity.
    \item \textbf{Algorithm Performance Analysis.} Classical query-complexity lower bounds often rely on change-of-measure arguments from sequential hypothesis testing \citep{lattimore2020bandit}. These arguments do not directly extend to the quantum setting, so new tools are needed to establish instance-dependent lower bounds for quantum algorithms.
\end{itemize}

To aid the reader, we provide some intuitive interpretation of the quantum concepts used in this paper. Classical MCTS typically uses a simulation model to generate stochastic rollouts and estimate the values of left nodes and candidate moves. In the quantum setting, the simulation model is accessed through a quantum oracle model, which serves as the quantum counterpart of a classical simulator. A query to a classical simulator at a leaf node returns a random observation, while a query to a quantum oracle prepares a quantum sample, which is a superposition state that encodes the possible outcomes. Measuring this state produces a classical sample but collapses the quantum state and destroys the quantum advantage. Accordingly, classical sample complexity is replaced by quantum query complexity as the main measure of computational cost. This restriction on information extraction motivates the \emph{lazy-measurement} principle used in our algorithm design.

In this work, we assume access to the quantum oracle model and focus on the design and analysis of MCTS algorithm. This assumption is natural when the underlying simulator is implemented by a classical program with accessible source code. Such a program can be represented as a Boolean circuit and implemented reversibly as a quantum circuit with polynomial overhead \citep{wang2021quantum2}. Similar assumptions are widely used in the quantum algorithm literature \citep{wang2021quantum, zhong2023provably,hu2025quantum}.

The main contributions of this paper are summarized as follows:

\begin{itemize}
    \item We formulate quantum MCTS in the fixed-confidence setting and propose a quantum MCTS algorithm (QMCTS) that combines quantum Monte Carlo (QMC) estimation with threshold-based elimination under a lazy-measurement principle. We prove that QMCTS identifies an $\epsilon$-optimal root move with probability at least $1-\delta$.

\item We establish an instance-dependent query-complexity upper bound for QMCTS and derive a lower bound for adaptive quantum algorithms. We characterize conditions under which the upper and lower bounds match up to logarithmic factors. Both analyses depend explicitly on the tree structure. 
The lower-bound analysis also introduces a new sequential quantum phase-testing result that may be useful beyond MCTS.

\item We conduct numerical experiments to verify the theoretical results and evaluate the performance of the proposed algorithms. We also test the methods on a large tree constructed from the Lichess database \citep{lichessOpenDatabase}. Based on this experiment, we develop a classical-quantum hybrid algorithm that combines the strengths of classical and quantum MCTS algorithms. Finally, we implement a small experiment on IBM quantum hardware to demonstrate the feasibility of QMCTS.
\end{itemize}

Our study is related to two strands of literature:

\textbf{Best Arm Identification.} Best arm identification (BAI) is one of the most extensively studied problem in the bandit literature \citep{kaufmann2016complexity,garivier2016optimal}. In the fixed-confidence setting, instance-dependent lower bounds and matching algorithms have been established \citep{degenne2019non,wang2021fast}. Recent work has also explored the use of quantum computing to improve BAI. \citet{casale2020quantum} first formulated quantum BAI based on amplitude amplification. \citet{wang2021quantum} established a quadratic speedup in query complexity under a strong quantum oracle model. More recently, \citet{wang2025quantum} studied BAI under quantum oracle models with different query structures. Our work differs from quantum BAI in two ways. First, we consider MCTS, where the goal is to identify the optimal move at the root, whose value is determined recursively by the MAX-MIN tree. This tree structure creates additional challenges for both algorithm design and query-complexity analysis. Second, we establish a new query-complexity lower bound for adaptive quantum MCTS algorithms. The BAI lower bound follows as a special case by considering a depth-one tree. 

\textbf{Monte Carlo Tree-Search.} Many classical MCTS methods, including UCT, are built on regret-minimizing bandit algorithms \citep{kocsis2006bandit}, while we consider a fixed-confidence formulation. \citet{kaufmann2017monte} studied fixed-confidence MCTS and developed BAI-based algorithms with instance-dependent guarantees. Our work instead studies how quantum computing can reduce the query complexity of MCTS.

Research on quantum acceleration for MCTS and related tree-search methods remains limited. \citet{peters2021quantum} applied quantum mean estimation to rollout values, while~\citet{sequeira2021quantum} studied quantum Sparse Sampling with coherent access to environment transitions. \citet{shukla2026coherent} constructed coherent rollout oracles to compare root actions through expected returns, and~\citet{finet2026quantum} used a neutral-atom subroutine in the expansion step of MCTS. These methods do not directly address our setting, where the tree is fixed, uncertainty arises from stochastic leaf values, and root decisions are determined by recursive MAX-MIN backups rather than rollout expectations. We further provide a complete query-complexity analysis, with instance-dependent upper and lower bounds for identifying an $\epsilon$-optimal root move.

\section{Quantum Monte Carlo Tree Search}
In this section, we first introduce the basic notation and concepts of quantum computing, and then formulate fixed-confidence MCTS under the quantum oracle model.

\subsection{Preliminaries on Quantum Computing}
\label{sec: preliminary}

Quantum computing exploits physical phenomena such as superposition, entanglement, and interference to process information. These mechanisms enable quantum parallelism and, when combined with suitable algorithmic procedures, can lead to significant speedups for certain computational tasks. We introduce the basic notation below and refer readers to \citep{nielsen2010quantum} for a comprehensive treatment. 

In quantum computing, Dirac's bra-ket notation is used to represent vectors and linear operations in a complex Hilbert space. A quantum state is represented by a \emph{ket} $|\psi\rangle\in\mathbb{C}^d$, and the corresponding \emph{bra}, denoted by $\langle\psi|$, is defined as the conjugate transpose of the ket, i.e., $\langle\psi|=|\psi\rangle^\dagger$. For two quantum states $|\psi\rangle$ and $|\phi\rangle$, the quantity $\langle\phi|\psi\rangle$ denotes their inner product, while $|\psi\rangle\langle\phi|$ denotes their outer product and defines a linear operator. Quantum operations are represented by unitary operators $U$ satisfying $U^\dagger U=UU^\dagger=I$, and their action on a quantum state is written as $|\psi\rangle\mapsto U|\psi\rangle$.

The basic unit of information in quantum computing is a \emph{qubit} (quantum bit). A classical bit takes one of two values, $0$ or $1$, while a qubit is described by the corresponding computational basis states $|0\rangle$ and $|1\rangle$. Unlike a classical bit, a qubit can exist in a superposition of these basis states and is generally written as
$
|\psi\rangle=\alpha|0\rangle+\beta|1\rangle,
$
where $\alpha,\beta\in\mathbb{C}$ are the corresponding probability amplitudes satisfying $|\alpha|^2+|\beta|^2=1$. Upon measurement, the qubit collapses to the state $|0\rangle$ with probability $|\alpha|^2$ and to the state $|1\rangle$ with probability $|\beta|^2$, respectively.

\subsection{Problem Formulation}

Consider a fixed game tree $\mathcal{T}$ with root node $s_0$. Each node is labeled as either MAX or MIN, with $s_0$ being a MAX node. For each node $s$, let $\mathcal C(s)$ denote the set of its children. Let $\mathcal L$ be the set of leaf nodes and $L:=|\mathcal L|$. Each leaf $\ell\in\mathcal L$ is associated with a stochastic oracle that represents the leaf evaluation when $\ell$ is reached by the MCTS algorithm. For ease of presentation, we consider Bernoulli leaf oracles with unknown means $\mu_\ell\in(0,1)$, where $\mu_\ell$ represents the probability that the first player wins the game from leaf $\ell$. {\color{black}Let $\mu := (\mu_\ell)_{\ell\in\mathcal L}\in(0,1)^L$ denote the vector of leaf means. The proposed algorithms extend directly to leaf distributions supported on $[0,1]$ (see Appendix \ref{sec: distribution} for details).}

The \emph{value} of each node $s$ (for the first player) is defined recursively as a function of $\mu$
\begin{equation}
V_s(\mu)=
\begin{cases}
\mu_s,
& s\in\mathcal{L},\\[0.2em]
\displaystyle\max_{c\in\mathcal{C}(s)} V_c(\mu),
& s\text{ is a MAX node},\\[0.2em]
\displaystyle\min_{c\in\mathcal{C}(s)} V_c(\mu),
& s\text{ is a MIN node}.
\end{cases}
\label{eq:value-recursion}
\end{equation}
We suppress the dependence on $\mu$ when it is clear from the context. The optimal root move is the child of $s_0$ with the highest value
\begin{equation}
s^\star
\in
\operatorname*{arg\,max}_{s\in\mathcal{C}(s_0)}
V_s.
\end{equation}
with ties resolved according to a fixed ordering.

To identify $s^\star$ or an $\epsilon$-optimal root move, an algorithm adaptively queries leaf nodes and observes their stochastic outcomes. After a random number $\tau$ of queries, the algorithm stops and recommends a move $\widehat{s}\in\mathcal C(s_0)$. Given a confidence level $\delta\in(0,1)$ and a precision parameter $\epsilon\in(0,1]$, the objective is to identify a move in $\mathcal C(s_0)$ whose value is within $\epsilon$ of the optimal root value with probability at least $1-\delta$, i.e.,
\begin{equation*}
    \mathbb{P}(V_{s_0} - V_{\widehat{s}}\le \epsilon)\ge 1-\delta,
\end{equation*}
while minimizing the total number of queries. We formally introduce the quantum oracle model in Definition~\ref{def: oracle}.

\begin{definition}[Quantum oracle model]
\label{def: oracle}
For each leaf $\ell\in\mathcal L$ with Bernoulli mean $\mu_\ell\in(0,1)$, we assume access to a unitary quantum oracle $U_\ell$ satisfying
\begin{equation}
U_{\ell}\lvert 0\rangle
=
\sqrt{1-\mu_{\ell}}\,
\lvert 0\rangle
+
\sqrt{\mu_{\ell}}\,
\lvert 1\rangle.
\label{eq:general-leaf-oracle}
\end{equation}
The algorithm may query $U_\ell$ or its inverse $U_\ell^\dagger$, and each call is counted as one quantum query.
\end{definition}
The quantum oracle model $U_\ell$ is the counterpart of a classical stochastic simulator. If $U_\ell|0\rangle$ is measured immediately, the outcome is $1$ with probability $\mu_\ell$ and $0$ with probability $1-\mu_\ell$, which reproduces one classical Bernoulli sample from leaf $\ell$. Before measurement, the quantum state can be processed coherently. 

We next introduce several quantities that characterize the difficulty of the tree-search problem and will be used in the query-complexity analysis. Define the root performance gap
\[
\Delta_\star = V_{s^\star} - \max_{s\in \mathcal C(s_0)\setminus \{s^\star\}}V_s \ge 0,
\]
which measures the difference between the optimal root move and the best competing move.
For an edge $(s,c)$ directed from a node $s$ to its child $c$, define the local edge gap
\[
g(s,c)=|V_s-V_c|
=
\begin{cases}
V_s-V_c, & s \text{ is a MAX node},\\
V_c-V_s, & s \text{ is a MIN node}.
\end{cases}
\]

For a leaf $\ell\in \mathcal L$, let $\text{path}(\ell)$ denote the set of edges on the path from the root to $\ell$, and define
\[
\Delta_\ell = \max_{(s,c)\in \text{path}(\ell)} g(s,c),
\]
which is the largest local edge gap along the path of $\ell$.

\section{Quantum Algorithm and Query Complexity Analysis}
In this section, we first develop the quantum MCTS algorithm and its classical counterpart. We then provide a complete query-complexity analysis to characterize the resulting quantum advantage.

\subsection{Quantum algorithm design}
To develop an efficient quantum algorithm for MCTS, we consider three key issues: (a) using quantum methods to improve the accuracy of leaf-value evaluation; (b) designing the algorithm under the lazy-measurement principle; and (c) exploiting the tree structure to design an effective elimination rule.
We first introduce the QMC estimator~\citep{kothari2023mean} in Lemma \ref{lem: QMC} and use it as a basic subroutine of QMCTS for leaf-value estimation. Compared with classical Monte Carlo method, QMC achieves a quadratic improvement in the dependence on the target accuracy while preserving the same confidence guarantee.

\begin{lemma}
\label{lem: QMC}
For each leaf $\ell\in\mathcal L$, there exists a quantum algorithm, denoted by
$\mathrm{QMC}(U_\ell,\alpha,\eta)$, that uses
$
O(\frac{1}{\alpha}\log\frac{1}{\eta})
$
queries to $U_\ell$ and returns an estimate $\widehat{\mu}_\ell$ satisfying
\[
\mathbb{P}\left(
    \left|\widehat{\mu}_\ell-\mu_\ell\right|>\alpha
\right)
\le \eta.
\]
\end{lemma}

However, naively replacing classical Monte Carlo estimates with QMC estimates is not sufficient to obtain a quantum speedup for tree search. The main difficulty comes from the restriction on information extraction in quantum computing. Classical sequential MCTS methods \citep{kaufmann2017monte} repeatedly use intermediate estimates to decide which nodes to sample or eliminate. In the quantum setting, obtaining such estimates requires measurement, which collapses the quantum state and may reduce the query advantage of QMC. On the other hand, delaying the elimination of clearly suboptimal moves can lead to many unnecessary queries. Therefore, a quantum MCTS algorithm needs to balance measurement frequency and adaptivity.

We design the quantum algorithm following a lazy-measurement principle. The algorithm proceeds in iterations, and the quantum state is measured only at the end of each iteration to obtain updated estimates. At iteration $r$, we first estimate the value of each active leaf and then propagate these estimates upward through the tree. Specifically, define
\begin{equation}
\widehat{V}_{s,r}
=
\begin{cases}
\widehat{\mu}_{s,r},
& s\in\mathcal{L}_{r-1},\\[0.3em]
\displaystyle
\max_{c\in\mathcal{C}_{r-1}(s)}
\widehat{V}_{c,r},
& s\text{ is a MAX node},\\[0.5em]
\displaystyle
\min_{c\in\mathcal{C}_{r-1}(s)}
\widehat{V}_{c,r},
& s\text{ is a MIN node},
\end{cases}
\label{eq:value-propagation}
\end{equation}
where $\mathcal{L}_{r-1}$ is the set of active leaves and
$\mathcal{C}_{r-1}(s)$ is the active children of node $s$ at the
beginning of iteration $r$. 

For an active edge $(s,c)$, define the empirical local gap by
\begin{equation}
\widehat{g}_r(s,c)
=
\begin{cases}
\widehat{V}_{s,r}-\widehat{V}_{c,r},
& s\text{ is a MAX node},\\[0.3em]
\widehat{V}_{c,r}-\widehat{V}_{s,r},
& s\text{ is a MIN node},
\end{cases}
\label{eq:empirical-local-gap}
\end{equation}
which measures how far the child $c$ is from the currently optimal child of node $s$. 

{\color{black}We use the geometrically decreasing threshold $\gamma_r = 2^{-r}$ to eliminate suboptimal nodes. In each iteration, the internal nodes are processed in reverse topological order. For each active internal node $s$, any child $c$ satisfying $\widehat{g}_r(s,c)>\gamma_r$ is removed together with all of its descendants. The remaining nodes then form the active subtree $\mathcal{T}_r$, with active leaf set $\mathcal{L}_r$. We set the QMC precision and confidence parameters to $\alpha_r=\gamma_r/2$ and $\eta_r=\delta/(2Lr^2)$, respectively. These choices allow clearly suboptimal nodes to be eliminated in early iterations without excessive query costs, while ensuring that all estimates used by the algorithm are accurate with probability at least $1-\delta$ (see Appendix \ref{app: query complexity} for details). This high-probability event is used to establish the correctness of the algorithm.} 

After the elimination step, let $\widehat{s}_r$ denote the remaining root child with the largest estimated value. The algorithm stops when either only one root move remains or $\gamma_r\le\epsilon$, and returns $\widehat{s}_r$. We call this algorithm QMCTS and summarize it in Algorithm~\ref{alg:qte-mcts}.

\begin{algorithm}[htpb]
\small
\caption{Quantum Monte Carlo Tree Search (QMCTS)}
\label{alg:qte-mcts}
\begin{algorithmic}[1]

\Require Finite MAX--MIN tree $\mathcal{T}$,
accuracy $\epsilon\in(0,1]$, and confidence
$\delta\in(0,1/2)$

\State $\mathcal{T}_0\gets\mathcal{T}$ and
$L\gets|\mathcal{L}|$

\For{$r=1,2,\ldots$}

    \State $\gamma_r\gets 2^{-r}$,
    $\alpha_r\gets\gamma_r/2$, and
    $\eta_r\gets\delta/(2Lr^2)$

    \ForAll{$\ell\in\mathcal{L}_{r-1}$}
        \State
        $\widehat{\mu}_{\ell,r}
        \gets\operatorname{QMC}(U_\ell,\alpha_r,\eta_r)$
    \EndFor

    \State Propagate $\widehat{V}_{s,r}$ upward using
    (\ref{eq:value-propagation})

    \ForAll{internal nodes $s$, in reverse topological order}
        \State Remove each child $c$ satisfying
        $\widehat{g}_r(s,c)>\gamma_r$
    \EndFor

    \State Let $\mathcal{T}_r$ be the remaining connected subtree,
    with active leaf set $\mathcal{L}_r$

    \State
    $\widehat{s}_r
    \in
    \operatorname*{arg\,max}_{s\in\mathcal{C}_r(s_0)}
    \widehat{V}_{s,r}$

    \If{$|\mathcal{C}_r(s_0)|=1$ \textbf{or}
        $\gamma_r\le\epsilon$}
        \State \Return $\widehat{s}_r$
    \EndIf
\EndFor
\end{algorithmic}
\end{algorithm}

For comparison, we also develop a classical counterpart of QMCTS, called CMCTS, summarized in Algorithm~\ref{alg:classical-reuse-mcts}. CMCTS follows the same elimination framework as QMCTS. Its main advantage is that samples collected in previous iterations can be reused when updating the leaf estimates. {\color{black}At each iteration, given the same precision $\alpha_r$ and confidence $\eta_r$, CMCTS uses Hoeffding's inequality to determine the 
the number of additional samples required for each active leaf.}

{\color{black}This motivates a Hybrid MCTS method that uses classical estimates until a fresh QMC estimate becomes more query-efficient than further classical sampling for each leaf. Since the query cost for each leaf depends only on the precision and confidence parameters, this comparison can be made before querying the oracles (see Algorithm \ref{alg:hybrid-mcts} for details).
The hybrid method avoids costly quantum estimation during coarse elimination while retaining the quantum query advantage at high precision.}

\begin{algorithm}[t]
\caption{Classical Monte Carlo Tree Search (CMCTS)}
\label{alg:classical-reuse-mcts}
\begin{algorithmic}[1]

\Require Finite MAX-MIN tree $\mathcal{T}$, precision
$\epsilon\in(0,1]$, and confidence $\delta\in(0,1/2)$

\State $\mathcal{T}_0\gets\mathcal{T}$ and
$L\gets|\mathcal{L}|$
\State $N_{\ell}\gets0$ for every $\ell\in\mathcal{L}$

\For{$r=1,2,\ldots$}

    \State $\gamma_r\gets2^{-r}$,
    $\alpha_r\gets\gamma_r/2$, and
    $\eta_r\gets\delta/(2Lr^2)$

    \State
    $n_r\gets
    \left\lceil
    \log(2/\eta_r)/(2\alpha_r^2)
    \right\rceil$

    \ForAll{$\ell\in\mathcal{L}_{r-1}$}

        \State
        $\Delta_{\ell}\gets
        \max\{0,n_r-N_{\ell}\}$

        \State Query leaf $\ell$ for
        $\Delta_{\ell}$ additional observations and retain them

        \State
        $N_{\ell}\gets N_{\ell}+\Delta_{\ell}$

        \State Construct $\widehat{\mu}^{\mathrm C}_{\ell,r}$
        as the empirical mean of all
        $N_{\ell}$ retained observations

    \EndFor

    \State Propagate $\widehat{V}^{\mathrm C}_{s,r}$ upward using (\ref{eq:value-propagation})

    \ForAll{internal nodes $s$, in reverse topological order}

        \State Remove each child $c$ satisfying
        $\widehat{g}^{\mathrm C}_r(s,c)>\gamma_r$

    \EndFor

    \State Let $\mathcal{T}_r$ be the remaining connected subtree,
    with active leaf set $\mathcal{L}_r$

    \State
    $\widehat{s}^{\mathrm C}_r
    \in
    \operatorname*{arg\,max}_{s\in\mathcal{C}_r(s_0)}
    \widehat{V}^{\mathrm C}_{s,r}$

    \If{$|\mathcal{C}_r(s_0)|=1$ \textbf{or}
        $\gamma_r\le\epsilon$}

        \State \Return $\widehat{s}^{\mathrm C}_r$

    \EndIf

\EndFor

\end{algorithmic}
\end{algorithm}

\subsection{Query Complexity Analysis}

In this subsection, we analyze the query complexity of QMCTS and CMCTS. We first derive the upper bounds.

For each leaf $\ell \in\mathcal L$ and iteration $r \ge 1$, define the bad estimation event
\begin{equation*}
F_{\ell, r} := \{\ell \in \mathcal L_{r-1}: |\widehat \mu_{\ell, r}-\mu_{\ell}|>\alpha_r \}.
\end{equation*}
If $\ell$ has already been eliminated, we set $F_{\ell,r}=\emptyset$. We then define the good event
\[
\mathcal G = \bigcap_{r\ge1}\bigcap_{\ell\in \mathcal L}F^c_{\ell, r},
\]
under which every estimate used by the algorithm is within its prescribed precision $\alpha_r$. Lemma~\ref{lem: good event} shows that the good event holds with probability at least $1-\delta$.

\begin{lemma}
\label{lem: good event}
The probability of the good event $\mathcal G$ is at least $1-\delta$, that is,
$\mathbb {P}(\mathcal G^c) \le \delta.
$
\end{lemma}

On the good event $\mathcal G$, we can further establish two auxiliary results, as detailed in Appendix \ref{app: query complexity}. These results show that the estimation error remains controlled as values are propagated up the tree, and that the elimination step removes only suboptimal nodes while preserving at least one optimal child at every node. 

Based on these results, Theorem~\ref{thm: corectness} shows that QMCTS identifies an $\epsilon$-optimal move with probability at least $1-\delta$.

\begin{theorem}
\label{thm: corectness}
    With probability at least $1-\delta$, QMCTS stops after finitely many iterations and returns an $\epsilon$-optimal move for $s_0$, i.e.,
    \begin{equation*}
    \mathbb{P}(V_{s_0} - V_{\widehat{s}}\le \epsilon)\ge 1-\delta.
\end{equation*}
\end{theorem}

Theorem~\ref{thm: query-complexity} establishes an upper bound on the query complexity of QMCTS.

\begin{theorem}
\label{thm: query-complexity}
With probability at least $1-\delta$, the total number of quantum
oracle queries used by QMCTS satisfies
\begin{equation}
\begin{aligned}
Q
=O\left(
\sum_{\ell\in\mathcal L}
\frac{1}{d_{\ell,\epsilon}}
\left[
\log\frac{2L}{\delta}
+
2\log\left(
2+\log_2\frac{1}{d_{\ell,\epsilon}}
\right)
\right]\right),
\end{aligned}
\label{eq:query-complexity-upper-bound}
\end{equation}
where
$
d_{\ell,\epsilon}
=
\Delta_\ell\vee\Delta_\star\vee\epsilon.
$
\end{theorem}

The query complexity is determined by the effective difficulty $d_{\ell,\epsilon}$ of each leaf. QMCTS scales linearly with the inverse difficulty, up to logarithmic factors. 

We also establishes the correctness of CMCTS and derive its query-complexity upper bound; see Theorems \ref{thm:cmcts-correctness} and \ref{thm:cmcts-query-complexity} in Appendix \ref{sec:CMCTS}. The classical query-complexity upper bound is
\begin{equation}
\label{eq: cmcts bound}
    O\left(
        \sum_{\ell\in\mathcal L}
        \frac{1}{d_{\ell,\epsilon}^2}
        \left[
            \log\frac{4L}{\delta}
            +
            2\log\left(
                2+\log_2\frac{1}{d_{\ell,\epsilon}}
            \right)
        \right]
    \right).
\end{equation}
Comparing (\ref{eq:query-complexity-upper-bound}) and (\ref{eq: cmcts bound}), QMCTS improves the dependence on the effective difficulty $d_{\ell,\epsilon}$ of each leaf from quadratic to linear. Thus, QMCTS achieves a quadratic speedup up to logarithmic factors.

We now turn to the query-complexity lower bound. Classical lower-bound analyses often rely on change-of-measure arguments from sequential hypothesis testing. These arguments do not directly extend to the quantum setting, so we develop a quantum counterpart.

\begin{theorem}[Informal]
\label{thm:informal-structural-quantum-lower}
Consider a problem instance $\mu$ for which the root decision depends on a
collection of decision-relevant leaves. For any $\delta\in(0,1/4]$, every
adaptive quantum algorithm that identifies an $\epsilon$-optimal root move
with probability at least $1-\delta$ requires
\[
    \mathbb E_\mu[\tau]
    =
    \Omega\left(
        \log\frac{1}{\delta}
        \sum_{\ell\in\mathcal L_{\mathrm{piv}}}
        \frac{1}{d_{\ell,\epsilon}}
    \right),
\]
oracle queries. Here, $\mathcal L_{\mathrm{piv}}$ denotes the set of pivotal leaves; see Definition~\ref{def:pivotal-block} in Appendix \ref{sec: lower bound}.
\end{theorem}

The formal version of Theorem \ref{thm:informal-structural-quantum-lower} is stated in Theorem \ref{thm:structural-quantum-lower} in Appendix \ref{sec: lower bound}, where we specify the structural conditions required for the lower-bound analysis. The lower bound depends only on pivotal leaves because, in a tree-search problem, not every leaf can affect the root decision. Hence, the query complexity is determined by the parts of the tree that matter for distinguishing the best move. The dependence on the inverse difficulty $d_{\ell,\epsilon}$ is linear, which shows that the quadratic speedup for each pivotal leaf is optimal up to logarithmic factors.

Corollary \ref{cor:matching-quantum-bounds} in the Appendix shows that, for uniformly pivotal instances with $\mathcal L_{\mathrm{piv}}=\mathcal L$, the upper and lower bounds match up to logarithmic factors. Thus, QMCTS is near-optimal in this case. This also highlights a structural difference between MCTS and BAI. In BAI, every arm can directly affect the final decision, which corresponds to the uniformly pivotal case. For a general MCTS problem instance, however, which leaves are pivotal is not known in advance. As a result, the algorithm may need to query leaves that eventually turn out to be nonpivotal. This extra cost comes from identifying the decision-relevant parts of the tree and makes the remaining gap between the upper and lower bounds difficult to close.

\textbf{Technical Novelty.} We conclude this section by highlighting the main technical ingredients in the query-complexity analysis. For the upper bound, we use induction on the tree to establish the validity of the elimination step and combine this argument with the QMC error guarantee to obtain the query-complexity bound. For the lower bound, we exploit the tree structure and introduce pivotal leaf blocks to characterize how individual leaves affect the root decision. {\color{black}We also develop a new sequential quantum phase-testing result (see Appendix \ref{sec: phase-testing} for details) and reduce the MCTS identification problem to a collection of quantum phase-testing problems.} This result is then used to derive the query-complexity lower bound and may be of independent interest for other quantum lower-bound analyses.

\section{Numerical Experiment}
In this section, we conduct numerical experiments to evaluate QMCTS. We use synthetic experiments to verify the theoretical results, a large-scale chess experiment to access the scalability of the algorithm on a practical problem, and a real quantum computer experiment to examine the feasibility of implementing QMCTS on a quantum processor.

\subsection{Synthetic Experiment}

We first examine how the query complexity scales with the inverse root gap $1/\Delta_\star$ and the inverse precision $1/\epsilon$. We consider a fixed depth-two MAX-MIN tree in which the MAX root has eight actions, each leading to a MIN node with two leaves, for a total of 25 nodes and 16 Bernoulli leaves. We set $\delta=0.05$ and run 100 independent replications for each setting. As a classical benchmark, we include UGapE-MCTS \citep{kaufmann2017monte}. In our implementation, QMCTS is simulated using standard quantum amplitude estimation as the QMC subroutine~\citep{brassard2000quantum}.

Figure~\ref{fig:complexity scaling} reports the mean query complexity as the problem difficulty increases. In all tested settings, all three algorithms terminate and return an $\epsilon$-optimal root action with a 100\% empirical success rate. Panel~(a) varies the inverse root gap with $\epsilon=2^{-15}$ fixed, while panel~(b) varies the inverse precision with $\Delta_\star=2^{-15}$ fixed. The fitted log-log slopes are $0.96$ and $0.95$ for QMCTS, $2.02$ and $2.01$ for CMCTS, and $1.95$ and $2.01$ for UGapE-MCTS, respectively. Thus, QMCTS exhibits approximately linear scaling with the inverse gap and inverse precision, while both classical methods exhibit approximately quadratic scaling. 
These results are consistent with the theoretical analysis and support the quantum advantage over the classical benchmarks. 

\begin{figure}[!t]
    \centering
    \includegraphics[width=1.0\linewidth]{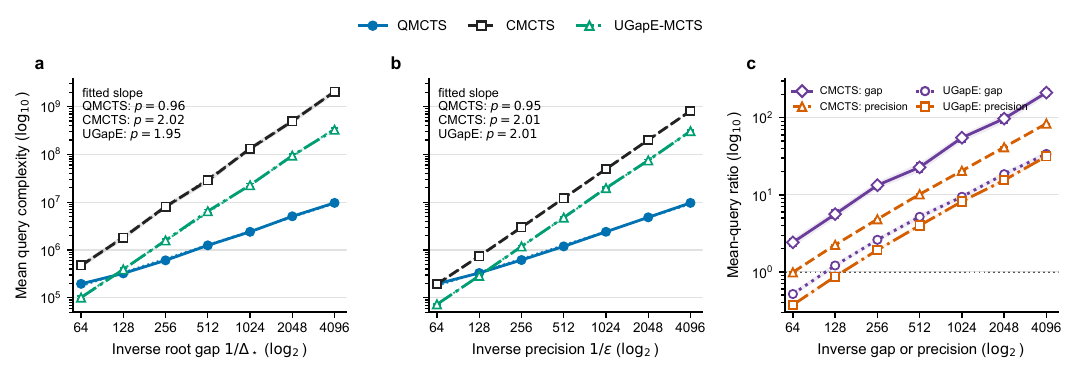}
    \caption{Query Complexity Scaling with Gap and Precision}
    \label{fig:complexity scaling}
\end{figure}

\subsection{Large-Scale Chess Experiment}

We next evaluate the algorithms on a large search tree constructed from the Lichess database~\citep{lichessOpenDatabase}. Historical games that follow the same move sequence from the selected root position are merged into a common path in the tree. Each leaf represents games that reach the same terminal position in the constructed tree. From the root player's perspective, wins, draws, and losses are assigned scores of $1$, $1/2$, and $0$, respectively, and the leaf value is defined as the empirical mean score of the associated games. A leaf-oracle query returns an independent Bernoulli reward with this mean. The experiment can be viewed as an offline policy-improvement problem in which historical outcomes define the stochastic leaf rewards.

We first compare QMCTS and CMCTS on a depth-$11$ tree with $9509$ nodes, $3072$ leaves, and 10 root moves. The root performance gap is $\Delta_\star=0.0020503$. We set $\delta=0.05$, vary $\epsilon$, and run $100$ independent replications for each setting.

On this large tree, pure QMCTS is less query-efficient when the inverse precision is small. In the early iterations, thousands of leaves remain active, so QMCTS repeatedly applies QMC to many leaves. In contrast, CMCTS requires only a small number of additional classical samples and reuses all previous observations. As a result, QMCTS incurs a large query cost before eliminating enough leaves. 

\begin{figure}[!t]
    \centering
    \includegraphics[width=1.0\linewidth]{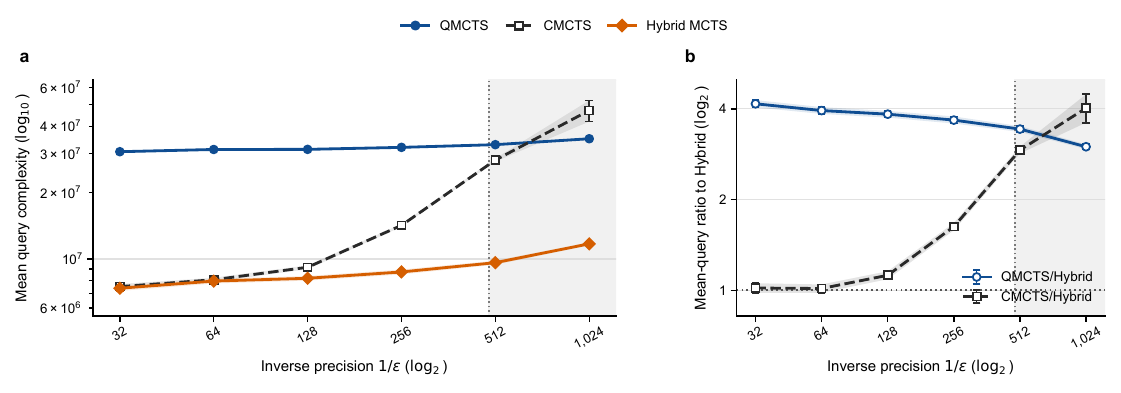}
    \caption{Large-scale chess query complexity.}
    \label{fig:large-scale chess}
\end{figure}

Figure~\ref{fig:large-scale chess} reports the mean query complexity in the large-scale chess experiment. All methods terminate and return an $\epsilon$-optimal root action with a $100\%$ empirical success rate. Panel~(a) shows that, at small inverse precision, CMCTS and Hybrid MCTS require substantially fewer queries than QMCTS. At $1/\epsilon=32$, the mean query complexities are $30.64$ million for QMCTS, $7.50$ million for CMCTS, and $7.37$ million for Hybrid MCTS. As the inverse precision increases, the query cost of CMCTS grows rapidly, while QMCTS benefits from its better scaling. At $1/\epsilon=1024$, the corresponding query complexities are $35.07$, $47.11$, and $11.71$ million. Panel~(b) reports the ratios relative to Hybrid MCTS. At $1/\epsilon=1024$, the QMCTS-to-Hybrid and CMCTS-to-Hybrid ratios are $2.994$ and $4.022$, respectively. These results show that Hybrid MCTS reduces the overall query cost by using classical sampling when the inverse precision is small and switching to quantum estimation as the inverse precision increases.

\subsection{Real Experiment}

We further evaluate QMCTS on an IBM Quantum processor using Qiskit~\citep{javadi2024quantum}. We consider a depth-two MAX--MIN tree with two root moves and four Bernoulli leaves, with $\Delta_\star=0.01640625$, $\epsilon=2^{-6}$, and $\delta=0.05$.

We run QMCTS and CMCTS independently five times. Both algorithms terminate and return the unique $\epsilon$-optimal root move in all replications. The mean query complexities are $68352$ for QMCTS and $306444$ for CMCTS. Thus, CMCTS uses about $4.48$ times more oracle queries, corresponding to a $77.70\%$ reduction for QMCTS. These results provide a small-scale hardware demonstration that QMCTS maintains correctness and a query-complexity advantage on a physical quantum processor.

\section{Conclusion and Limitations}
In this paper, we study how quantum computing can improve MCTS. We develop a quantum algorithm, QMCTS, and establish its correctness and query-complexity guarantees. Numerical experiments show a clear query advantage over classical benchmarks. The main limitations are twofold. First, the gap between the upper and lower query-complexity bounds remains open for general tree instances, and we conjecture that closing this gap may be difficult because of the tree structure of the problem.  Second, query complexity does not fully capture the gate complexity and implementation cost on real quantum hardware.
Future work may extend the framework to more general tree-search models, including broader leaf-value distributions, stochastic transitions, and additional sources of uncertainty, and develop hardware-level complexity analyses that account for circuit depth and gate complexity.

\bibliographystyle{informs2014} 
\bibliography{sample} 
\ECSwitch



%
%
%

\section{Query Complexity Analysis}
\label{app: query complexity}

\textit{Proof of Lemma \ref{lem: good event}.}
Let $\mathcal F_{r-1}$ be the sigma-field generated by all operations and measurement
outcomes before iteration $r$. Condition on $\mathcal F_{r-1}$, according to the Lemma \ref{lem: QMC}, the estimators satisfy
\[
\mathbb P(F_{\ell, r}|\mathcal F_{r-1}) \le \mathbb I(\ell \in \mathcal L_{r-1})\eta_r \le \eta_r.
\]
Therefore, a union bound gives
\begin{equation*}
\begin{aligned}
    \mathbb P(\mathcal G^c) \le \sum_{r\ge 1}\sum_{\ell \in\mathcal L} \eta_r = \frac{\delta}{2}\sum_{r\ge 1}\sum_{\ell \in\mathcal L} \frac{1}{Lr^2} = \frac{\pi^2}{12}\delta < \delta.
\end{aligned}
\end{equation*}

For an active node before iteration $r$, let $V^{r-1}_s$ be its true value computed on the active subtree $\mathcal T_{r-1}$.

\begin{lemma}
\label{lem:value-estimation}
On the good event $\mathcal G$, for each iteration $r$ and every node
$s$ active before iteration $r$,
\[
    \left|\widehat{V}_{s,r}-V_s^{(r-1)}\right|
    \le \alpha_r .
\]
\end{lemma}

Fix an iteration $r$. We prove the claim by induction on the height of
$s$ in the active tree $\mathcal T_{r-1}$, where the height of a node
is the maximum number of edges from that node to an active descendant
leaf.

For an active leaf $\ell\in\mathcal L_{r-1}$, we have
$
\widehat V_{\ell,r}=\widehat\mu_{\ell,r},
$ and $
V_{\ell}^{(r-1)}=\mu_\ell.
$
By the definition of the good event $\mathcal G$,
\[
\left|\widehat V_{\ell,r}-V_{\ell}^{(r-1)}\right|
\le \alpha_r.
\]

Now consider an active internal node $s$. Assume that the claim holds
for every active child
$c\in\mathcal C_{r-1}(s)$. 

Suppose first that $s$ is a MAX node. By the recursive definitions of
the empirical and true values,
\[
\widehat V_{s,r}
=
\max_{c\in\mathcal C_{r-1}(s)}
\widehat V_{c,r},
\qquad
V_s^{(r-1)}
=
\max_{c\in\mathcal C_{r-1}(s)}
V_c^{(r-1)}.
\]
The induction hypothesis implies that, for every active child $c$,
\[
V_c^{(r-1)}-\alpha_r
\le
\widehat V_{c,r}
\le
V_c^{(r-1)}+\alpha_r.
\]
Taking the maximum over
$c\in\mathcal C_{r-1}(s)$ yields
\[
V_s^{(r-1)}-\alpha_r
\le
\widehat V_{s,r}
\le
V_s^{(r-1)}+\alpha_r.
\]
Therefore,
\[
\left|\widehat V_{s,r}-V_s^{(r-1)}\right|
\le \alpha_r.
\]

If $s$ is a MIN node, then
\[
\widehat V_{s,r}
=
\min_{c\in\mathcal C_{r-1}(s)}
\widehat V_{c,r},
\qquad
V_s^{(r-1)}
=
\min_{c\in\mathcal C_{r-1}(s)}
V_c^{(r-1)}.
\]
Applying the same inequalities and taking the minimum gives
\[
V_s^{(r-1)}-\alpha_r
\le
\widehat V_{s,r}
\le
V_s^{(r-1)}+\alpha_r,
\]
and thus
\[
\left|\widehat V_{s,r}-V_s^{(r-1)}\right|
\le \alpha_r.
\]

By induction, the claim holds for every node active before iteration
$r$.

\begin{lemma}
\label{lem:safe-threshold}
On the good event $\mathcal G$, every deleted child is strictly
suboptimal for its parent in the active true tree. Moreover, at least
one optimal child of every active internal node is retained.
Consequently, after every iteration,
$
V_s^{(r)}=V_s^{(r-1)}=V_s
$
for every retained node $s$.
\end{lemma}

We proceed by induction over the iterations. Before the first
iteration, the active tree is the original tree, so the result holds.

Fix iteration $r$ and suppose that
$
V_s^{(r-1)}=V_s
$
for every node active before this iteration.

Consider first a MAX node $s$. If a child $c$ is deleted, then
$
\widehat V_{s,r}-\widehat V_{c,r}>\gamma_r.
$
By Lemma~\ref{lem:value-estimation} and
$2\alpha_r=\gamma_r$,
\[
\begin{aligned}
V_s^{(r-1)}-V_c^{(r-1)}
&\ge
\widehat V_{s,r}-\widehat V_{c,r}-2\alpha_r>
\gamma_r-2\alpha_r
=0.
\end{aligned}
\]
Hence, $c$ is strictly suboptimal for $s$.

Let $c^\star$ be an optimal child of $s$, so that
$
V_{c^\star}^{(r-1)}=V_s^{(r-1)}.
$
Then
\[
\widehat V_{s,r}-\widehat V_{c^\star,r}
\le
2\alpha_r
=
\gamma_r,
\]
so $c^\star$ is not deleted.

The argument for a MIN node is identical, with the inequalities
reversed. Therefore, every active internal node retains at least one
child attaining its true MAX or MIN value.

Applying this argument from the deepest internal nodes toward the
root, we have
$
V_s^{(r)}=V_s^{(r-1)}=V_s.
$
The conclusion follows by induction over $r$.

\subsection{Proof of Theorem \ref{thm: corectness}}

Because $\gamma_r=2^{-r}$, we have $\gamma_r\to 0$ as
$r\to\infty$. Since $\epsilon>0$, there exists a finite iteration
\[
R_\epsilon
:=
\min\left\{
r\ge 1:\gamma_r\le\epsilon
\right\}.
\]
The stopping condition of Algorithm~\ref{alg:qte-mcts} is
satisfied no later than iteration $R_\epsilon$. Hence, the algorithm
terminates after finitely many iterations.

Condition on the good event $\mathcal G$. By
Lemma~\ref{lem:safe-threshold}, threshold elimination preserves the
true value of every retained node. Moreover, at least one optimal root
move remains active after every iteration. Consequently, for every
iteration $r$, there exists an active root move
$
s_r^\star\in\mathcal C_r(s_0)
$
such that
$
V_{s_r^\star}=V_{s_0}.
$

Consider the two possible stopping
conditions.
First, suppose that
$
\left|\mathcal C_\tau(s_0)\right|=1.
$
Since at least one optimal root move remains active, the unique
remaining move must be optimal. Therefore,
\[
V_{s_0}-V_{\widehat{s}}=0\le\epsilon.
\]

Now suppose that the algorithm stops because
$
\gamma_\tau\le\epsilon.
$
The returned move satisfies
\[
\widehat{s}
\in
\operatorname*{arg\,max}_{s\in\mathcal C_\tau(s_0)}
\widehat V_{s,\tau}.
\]
Let $s_\tau^\star$ be an optimal active root move. Then
$
\widehat V_{\widehat{s},\tau}
\ge
\widehat V_{s_\tau^\star,\tau}.
$
By Lemma~\ref{lem:value-estimation}, and because the true values of
retained nodes are preserved by Lemma~\ref{lem:safe-threshold},
\[
\left|
\widehat V_{s,\tau}-V_s
\right|
\le\alpha_\tau
\]
for every root move $s$ active at iteration $\tau$. Hence,
\begin{align*}
V_{s_0}-V_{\widehat{s}}
&=
V_{s_\tau^\star}-V_{\widehat{s}}\\
&=
\left(
V_{s_\tau^\star}
-
\widehat V_{s_\tau^\star,\tau}
\right)
+
\left(
\widehat V_{s_\tau^\star,\tau}
-
\widehat V_{\widehat{s},\tau}
\right)
+
\left(
\widehat V_{\widehat{s},\tau}
-
V_{\widehat{s}}
\right)\\
&\le
\alpha_\tau+0+\alpha_\tau\\
&=
2\alpha_\tau\\
&\le
\gamma_\tau\\
&\le
\epsilon.
\end{align*}

Thus, on the event $\mathcal G$, the algorithm returns an
$\epsilon$-optimal root move. Since
$
\mathbb P(\mathcal G)\ge1-\delta,
$
we conclude that
\[
\mathbb P\!\left(
V_{s_0}-V_{\widehat{s}}\le\epsilon
\right)
\ge
\mathbb P(\mathcal G)
\ge
1-\delta.
\]

\subsection{Proof of Theorem \ref{thm: query-complexity}}

We prove the query bound on the good event $\mathcal G$.

For each leaf $\ell\in\mathcal L$, let $R_\ell$ denote the last
iteration in which $\ell$ is queried. 
Suppose first that $R_\ell\ge 2$. Since $\ell$ is queried in iteration
$R_\ell$, it must remain active after iteration $R_\ell-1$ and
iteration $R_\ell-1$ is nonterminal.

Consider any edge $(s,c)$ on the active path from the root to
$\ell$. Since this edge survives iteration $R_\ell-1$, its empirical
gap satisfies
$
\widehat g_{R_\ell-1}(s,c)
\le
\gamma_{R_\ell-1}.
$
By Lemma \ref{lem:value-estimation} and Lemma \ref{lem:safe-threshold}
\[
\left|
\widehat V_{u,R_\ell-1}-V_u
\right|
\le
\alpha_{R_\ell-1}
\]
for every active node $u$. Therefore,
\[
g(s,c)
\le
\widehat g_{R_\ell-1}(s,c)
+
2\alpha_{R_\ell-1}.
\]
Since $2\alpha_r=\gamma_r$, it follows that
$
g(s,c)
\le
2\gamma_{R_\ell-1}.
$
Taking the maximum over all edges on the path to $\ell$ gives
$
\Delta_\ell
\le
2\gamma_{R_\ell-1}.
$

Because iteration $R_\ell-1$ is nonterminal, we have
$
\epsilon<\gamma_{R_\ell-1}
\le
2\gamma_{R_\ell-1}.
$

Furthermore, at least two root moves remain after iteration
$R_\ell-1$. By Lemma \ref{lem:safe-threshold}, at least one
optimal root move remains active. Let $s^\star$ be such an optimal
move and let $s\neq s^\star$ be another active root move. Since the
edge $(s_0,s)$ has not been eliminated,
\[
V_{s^\star}-V_s
\le
2\gamma_{R_\ell-1}.
\]
By the definition of the root gap,
\[
\Delta_\star
\le
V_{s^\star}-V_s
\le
2\gamma_{R_\ell-1}.
\]

Combining these inequalities yields
\[
d_{\ell,\epsilon}
=
\Delta_\ell\vee\Delta_\star\vee\epsilon
\le
2\gamma_{R_\ell-1}.
\]
Since $\gamma_r=2^{-r}$,
$
\gamma_{R_\ell-1}
=
2\gamma_{R_\ell},
$
and thus
\begin{equation}
d_{\ell,\epsilon}
\le
4\gamma_{R_\ell}.
\label{eq:last-query-resolution}
\end{equation}
Equivalently,
\begin{equation}
\frac{1}{\gamma_{R_\ell}}
\le
\frac{4}{d_{\ell,\epsilon}}.
\label{eq:last-query-threshold-bound}
\end{equation}

If $R_\ell=1$, then
\[
\frac{1}{\gamma_{R_\ell}}
=
2
\le
\frac{4}{d_{\ell,\epsilon}},
\]
because $d_{\ell,\epsilon}\le 1$. Hence,
\eqref{eq:last-query-threshold-bound} holds for every leaf.

Equation~\eqref{eq:last-query-resolution} also implies
$
d_{\ell,\epsilon}
\le
4\cdot 2^{-R_\ell},
$
so
\begin{equation}
R_\ell
\le
2+\log_2\frac{1}{d_{\ell,\epsilon}}.
\label{eq:last-query-round-bound}
\end{equation}

We now bound the queries allocated to leaf $\ell$. In iteration $r$, according to Lemma \ref{lem: QMC},
QMC uses at most
$
\frac{C_{\mathrm Q}}{\alpha_r}
\log\frac{1}{\eta_r}
$
queries to $U_\ell$ and $U_\ell^\dagger$, where $C_Q>0$ is constant. Since
$
\alpha_r={\gamma_r}/{2},
\eta_r={\delta}/{2Lr^2},
$
the total query complexity associated with $\ell$ satisfies
\begin{align*}
Q_\ell
&\le
C_{\mathrm Q}
\sum_{r=1}^{R_\ell}
\frac{1}{\alpha_r}
\log\frac{1}{\eta_r}=
2C_{\mathrm Q}
\sum_{r=1}^{R_\ell}
\frac{1}{\gamma_r}
\left[
\log\frac{2L}{\delta}
+
2\log r
\right].
\end{align*}
Then, we have
\begin{align*}
Q_\ell
&\le
2C_{\mathrm Q}
\left[
\log\frac{2L}{\delta}
+
2\log R_\ell
\right]
\sum_{r=1}^{R_\ell}\frac{1}{\gamma_r}.
\end{align*}
Since $\gamma_r=2^{-r}$,
\[
\sum_{r=1}^{R_\ell}\frac{1}{\gamma_r}
=
\sum_{r=1}^{R_\ell}2^r
<
2^{R_\ell+1}
=
\frac{2}{\gamma_{R_\ell}}.
\]
Consequently,
\[
Q_\ell
<
\frac{4C_{\mathrm Q}}{\gamma_{R_\ell}}
\left[
\log\frac{2L}{\delta}
+
2\log R_\ell
\right].
\]

Applying
\eqref{eq:last-query-threshold-bound} and
\eqref{eq:last-query-round-bound} gives
\[
Q_\ell
\le
\frac{16C_{\mathrm Q}}{d_{\ell,\epsilon}}
\left[
\log\frac{2L}{\delta}
+
2\log\left(
2+\log_2\frac{1}{d_{\ell,\epsilon}}
\right)
\right].
\]

Finally, every oracle query is associated with exactly one active
leaf. Summing the preceding inequality over
$\ell\in\mathcal L$ proves
\eqref{eq:query-complexity-upper-bound}.

\section{Query Complexity Lower Bound}
\label{sec: lower bound}

\subsection{Sequential Quantum Phase Hypothesis Testing}
\label{sec: phase-testing}
In this subsection, we present the quantum-information-theoretic tool used in
the lower-bound analysis. This is independent of the tree
search problem and considers only a sequential hypothesis-testing problem
for an unknown phase.

Let
\[
P_\theta
=
|0\rangle\langle 0|
+
e^{i\theta}|1\rangle\langle 1|,
\qquad
\theta\in[-\pi,\pi),
\]
be a one-qubit phase oracle. Fix a $\theta_{\mathrm c}$ and a measurable set
$\Theta\subseteq[-\pi,\pi)$ containing $\theta_{\mathrm c}$.
We consider the testing problem
\[
H_0:\theta=\theta_{\mathrm c}
\qquad\text{vs}\qquad
H_1:\theta\in[-\pi,\pi)\setminus\Theta.
\]
The set $\Theta$ plays the role of an \emph{indifference region}:
when
$
\theta\in\Theta\setminus\{\theta_{\mathrm c}\},
$
the testing procedure is not required to make either decision
correctly. The procedure must distinguish $\theta_{\mathrm c}$ from every phase outside
$\Theta$.

Let $\operatorname{Leb}$ denote the Lebesgue measure on $\mathbb{R}$ (i.e., length). Define
$
s:=\operatorname{Leb}(\Theta).
$
A smaller $s$ places admissible alternatives closer to
the reference phase and makes the two hypotheses harder to distinguish.
For example, if
$
\Theta
=
[\theta_{\mathrm c}-h,\theta_{\mathrm c}+h],
$
then the test is required to distinguish $\theta_{\mathrm c}$ from
all phases satisfying
$
|\theta-\theta_{\mathrm c}|>h,
$
and $s=2h$.

Consider an arbitrary sequential quantum procedure.  Between oracle
queries, the procedure may apply arbitrary $\theta$-independent
quantum operations and measurements, and its subsequent operations may
depend on previous measurement outcomes.  The procedure may query
$P_\theta$, $P_\theta^\dagger$ and stops
at a random time.  Upon stopping, it outputs
$D\in\{0,1\}$, where $D=0$ corresponds to accepting $H_0$ and $D=1$
corresponds to accepting $H_1$.

For a prescribed error level $\zeta\in(0,1/4]$, we require
$
\mathbb P_{\theta_{\mathrm c}}(D=1)\le\zeta
$
and
$
\sup_{\theta\notin\Theta}
\mathbb P_\theta(D=0)\le\zeta.
$

Let $N$ denote the random number of phase-oracle queries made by the
procedure before it stops. The quantity
$
\mathbb E_{\theta_{\mathrm c}}[N]
$
is therefore the average number of phase-oracle queries used when the
true phase is $\theta_{\mathrm c}$.  

\begin{lemma}
\label{lem:phase-testing}
Under the testing set above, if
$s\in(0, \frac{\pi}{2}],
$
then
\[
\mathbb E_{\theta_{\mathrm c}}[N]
\ge
\frac{1}{16s}\log\frac{1}{4\zeta}
-\frac12.
\]
\end{lemma}

If $\zeta>1/4$, then
\[
\frac{1}{16s}\log\frac{1}{4\zeta}-\frac12<0,
\]
and the claimed bound follows immediately.

For a fixed integer $n\ge 0$, define
$
a_n(\theta)
:=
\mathbb P_\theta(D=0,\;N\le n),
$
which is the acceptance probability of the test truncated after $n$ queries.

By the polynomial method for quantum phase queries,
$a_n(\theta)$ is a real trigonometric polynomial of degree at most
$2n$; see, e.g., \citet{beals2001quantum,mande2026tight}. Moreover, for every
$\theta\notin\Theta$,
\[
a_n(\theta)
\le
\mathbb P_\theta(D=0)
\le
\zeta.
\]
Since the set $\Theta$ has total angular length $s$, a
trigonometric Remez inequality gives
$
a_n(\theta_{\mathrm c})
\le
\zeta e^{8ns};
$
see, e.g., \citet{borwein2012polynomials,mande2026tight}.

Choose
\[
n
=
\left\lfloor
\frac{1}{8s}\log\frac{1}{4\zeta}
\right\rfloor.
\]
Then
$
\zeta e^{8ns}\le \frac14,
$
and hence
$
a_n(\theta_{\mathrm c})\le \frac14.
$
On the other hand, correctness under the null hypothesis implies
$
\mathbb P_{\theta_{\mathrm c}}(D=0)\ge 1-\zeta.
$
Therefore,
\begin{align*}
\mathbb P_{\theta_{\mathrm c}}(N>n)
&\ge
\mathbb P_{\theta_{\mathrm c}}(D=0,\;N>n)\\
&=
\mathbb P_{\theta_{\mathrm c}}(D=0)
-
a_n(\theta_{\mathrm c})\\
&\ge
1-\zeta-\frac14\\
&\ge
\frac12,
\end{align*}
where the last inequality uses $\zeta\le 1/4$.

Consequently,
\[
\mathbb E_{\theta_{\mathrm c}}[N]
\ge
n\,\mathbb P_{\theta_{\mathrm c}}(N>n)
\ge
\frac n2.
\]
Finally, using $\lfloor x\rfloor\ge x-1$,
\[
\mathbb E_{\theta_{\mathrm c}}[N]
\ge
\frac{1}{16s}\log\frac{1}{4\zeta}
-\frac12.
\]

\subsection{Lower Bound Analysis}
In this subection, we reduce the quantum MCTS identification problem to a sequential quantum phase hypothesis-testing problem and use the Lemma \ref{lem:phase-testing} to derive the query complexity lower bound.

For any leaf-mean vector $\nu\in(0,1)^{L}$, define the set of $\epsilon$-optimal moves at the root $s_0$ by
\begin{equation}
    \mathcal M_\epsilon(\nu)
    :=
    \left\{
        s\in\mathcal C(s_0):
        V_{s_0}(\nu)-V_s(\nu)\le\epsilon
    \right\}.
\label{eq:epsilon-optimal-root-set}
\end{equation}

The reduction starts by considering a family of leaf-mean vectors
$
    \widetilde\mu(\theta),
    \theta\in[-\pi,\pi),
$
constructed so that, for some reference phase $\theta_{\mathrm c}$,
\begin{equation}
    \widetilde\mu(\theta_{\mathrm c})=\mu.
\label{eq:phase-family-center}
\end{equation}
If
\begin{equation}
    \mathcal M_\epsilon(\mu)
    \cap
    \mathcal M_\epsilon\bigl(\widetilde\mu(\theta)\bigr)
    = \emptyset,
\label{eq:decision-changing-instance}
\end{equation}
then no root move is simultaneously a valid recommendation for the
original instance and the perturbed instance. Therefore, any algorithm
that is $(\epsilon,\delta)$-correct on both instances must distinguish
between them.

Define
\begin{equation}
    \Theta
    :=
    \left\{
        \theta\in[-\pi,\pi):
        \mathcal M_\epsilon(\mu)
        \cap
        \mathcal M_\epsilon\bigl(\widetilde\mu(\theta)\bigr)
        \neq \emptyset
    \right\}.
\label{eq:tree-induced-indifference}
\end{equation}
At $\theta=\theta_{\mathrm c}$ the instance is $\mu$, while every
$\theta\notin\Theta$ produces an instance whose valid root
recommendations are disjoint from those of $\mu$.  Consequently, a
correct tree-search algorithm can be used to test
\[
    H_0:\theta=\theta_{\mathrm c}
    \qquad\text{vs}\qquad
    H_1:\theta\notin\Theta.
\]

In the Monte Carlo tree search problem, changing the mean of a single leaf does not necessarily change the optimal move at the root, because its effect may be masked by the MAX--MIN operations along the path to the root. Therefore, we group leaf nodes into blocks in Definition \ref{def:pivotal-block} and consider joint perturbations of the leaves within each block, so that the perturbation can change the set of $\epsilon$-optimal moves at the root. 

\begin{definition}
\label{def:pivotal-block}
Fix a tree instance
$
    \mu=(\mu_\ell)_{\ell\in\mathcal L},
$
an accuracy level $\epsilon\in(0,1]$, and a constant $C>0$.
A collection of nonempty leaf sets
$
    \mathcal B_1,\ldots,\mathcal B_m\subseteq\mathcal L
$
is called a \emph{disjoint $C$-pivotal-block family} at $\mu$ if
the blocks are pairwise disjoint and there exists a common reference
phase $\theta_{\mathrm c}$ such that, for every block $\mathcal B_j$,
there exists a valid leaf-mean vectors
\[
    \mu^{(j)}(\theta)
    =
    \bigl(\mu^{(j)}_\ell(\theta)\bigr)_{\ell\in\mathcal L},
    \qquad
    \theta\in[-\pi,\pi),
\]
satisfying the following properties.

\begin{enumerate}

    \item
    Under $\theta_c$, the perturbed instance is exactly the
    original instance:
    $
        \mu^{(j)}(\theta_{\mathrm c})=\mu.
    $
    For an arbitrary phase $\theta$, every leaf outside the block
    remains unchanged:
    $
        \mu^{(j)}_\ell(\theta)
        =
        \mu_\ell,
        \forall\,\ell\notin\mathcal B_j.
    $
    Hence, when $\theta\neq\theta_{\mathrm c}$, only the means of
    leaves in $\mathcal B_j$ are allowed to differ from those of the
    original instance. 

    \item
    For every $\ell\in\mathcal B_j$, let
    $U^{(j)}_\ell(\theta)$ denote the quantum leaf oracle
    corresponding to the Bernoulli mean
    $\mu^{(j)}_\ell(\theta)$.
    One call to $U^{(j)}_\ell(\theta)$, its inverse, or a controlled
    version can be implemented using at most one call to
    $P_\theta$, its inverse, or a controlled version, together with
    quantum operations independent of $\theta$.

    \item
    Define the decision-preserving region
    \[
        \Theta_j
        :=
        \left\{
            \theta\in[-\pi,\pi):
            \mathcal M_\epsilon(\mu)
            \cap
            \mathcal M_\epsilon\bigl(\mu^{(j)}(\theta)\bigr)
            \neq\emptyset
        \right\},
    \]
    and let
    $
        d_{\mathcal B_j}
        :=
        \min_{\ell\in\mathcal B_j}
        d_{\ell,\epsilon}.
    $
    The perturbation family satisfies
    \[
        0<
        \operatorname{Leb}(\Theta_j)
        \le
        \min\left\{
            \frac{\pi}{2},
            C d_{\mathcal B_j}
        \right\}.
    \]

\end{enumerate}
\end{definition}

Let
\[
    \mathcal L_{\mathrm{piv}}
    :=
    \bigcup_{j=1}^{m}\mathcal B_j,
    \qquad
    b
    :=
    \max_{1\le j\le m}|\mathcal B_j|.
\]
For $x\in\mathbb R$, write
$
    [x]_+:=\max\{x,0\}.
$

\begin{theorem}
\label{thm:structural-quantum-lower}
Suppose that $\mu$ admits a disjoint $C$-pivotal-block family with $C$ and $b$ bounded by universal constants. For $\delta\in(0,1/4]$, the total number of leaf-oracle
queries $\tau$ satisfies
\[
    \mathbb E_\mu[\tau]
    =
    \Omega\left(
        \log\frac{1}{\delta}
        \sum_{\ell\in\mathcal L_{\mathrm{piv}}}
        \frac{1}{d_{\ell,\epsilon}}
    \right).
\]
\end{theorem}

For each leaf $\ell$, let $N_\ell$ denote the number of queries to its
oracle before the algorithm stops. For block $\mathcal B_j$, define
$
    N_j
    :=
    \sum_{\ell\in\mathcal B_j}N_\ell.
$

Fix one block $\mathcal B_j$ and run the tree algorithm on the family
$\mu^{(j)}(\theta)$ from Definition~\ref{def:pivotal-block}. Convert
its recommendation $\widehat{s}$ into a binary decision by setting
\[
    D_j=0
    \quad\Longleftrightarrow\quad
    \widehat{s}\in\mathcal M_\epsilon(\mu).
\]

At $\theta=\theta_{\mathrm c}$, $(\epsilon,\delta)$-correctness gives
$
    \mathbb P_{\theta_{\mathrm c}}(D_j=1)
    \le
    \delta.
$

For every $\theta\notin\Theta_j$, the sets
$\mathcal M_\epsilon(\mu)$ and $\mathcal M_\epsilon(\mu^{(j)}(\theta))
$
are disjoint. Therefore, whenever the algorithm is correct on the
perturbed instance, its recommendation cannot belong to
$\mathcal M_\epsilon(\mu)$, and
$
    \sup_{\theta\notin\Theta_j}
    \mathbb P_\theta(D_j=0)
    \le
    \delta.
$
Thus, the tree algorithm induces the phase test of
Lemma~\ref{lem:phase-testing}.

By Definition~\ref{def:pivotal-block}, the induced test can be
simulated using no more than $N_j$ phase-oracle queries. Therefore,
Lemma~\ref{lem:phase-testing} and
$
    \operatorname{Leb}(\Theta_j)
    \le
    C d_{\mathcal B_j}
$
imply
\begin{equation}
    \mathbb E_\mu[N_j]
    \ge
    \frac{1}{16C d_{\mathcal B_j}}
    \log\frac{1}{4\delta}
    -
    \frac12,
\label{eq:block-lower-in-d}
\end{equation}
whenever $\delta\le1/4$.
If $\delta>1/4$, then
$
    \log\frac{1}{4\delta}<0.
$
The theorem then
follows immediately. 

Since
$
    d_{\mathcal B_j}
    =
    \min_{\ell\in\mathcal B_j}
    d_{\ell,\epsilon},
$
we have
\begin{align*}
    \frac{1}{d_{\mathcal B_j}}
    &=
    \max_{\ell\in\mathcal B_j}
    \frac{1}{d_{\ell,\epsilon}}
    \ge
    \frac{1}{|\mathcal B_j|}
    \sum_{\ell\in\mathcal B_j}
    \frac{1}{d_{\ell,\epsilon}}
    \ge
    \frac{1}{b}
    \sum_{\ell\in\mathcal B_j}
    \frac{1}{d_{\ell,\epsilon}}.
\end{align*}

Hence,
$
    \tau
    \ge
    \sum_{j=1}^{m}N_j.
$
Taking expectations, summing
\eqref{eq:block-lower-in-d} over the blocks, and applying the preceding
inequality gives
\[
    \mathbb E_\mu[\tau]
    \ge
    \frac{1}{16bC}
    \log\frac{1}{4\delta}
    \sum_{\ell\in\mathcal L_{\mathrm{piv}}}
    \frac{1}{d_{\ell,\epsilon}}
    -
    \frac{m}{2}.
\]

Define 
\[
    H_{\mathrm Q}(\mu,\epsilon)
    :=
    \sum_{\ell\in\mathcal L}
    \frac{1}{d_{\ell,\epsilon}}.
\]
We call an tree instance \emph{uniformly pivotal} if it admits a disjoint
pivotal-block family that covers $\mathcal L$ and whose perturbation
constant $C$ and maximum block size $b$ are bounded by universal constants.
We use $\widetilde O$ to suppress logarithmic factors.

\begin{corollary}
\label{cor:matching-quantum-bounds}
On every uniformly pivotal instance, any
$(\epsilon,\delta)$-correct quantum algorithm satisfies
\[
    \mathbb E_\mu[\tau]
    =
    \Omega\left(
        H_{\mathrm Q}(\mu,\epsilon)
        \log\frac{1}{\delta}
    \right),
\]
whereas QMCTS uses
\[
    \widetilde O\left(
        H_{\mathrm Q}(\mu,\epsilon)
        \log\frac{1}{\delta}
    \right)
\]
queries with probability at least $1-\delta$. Therefore, the expected lower bound and the high-probability upper bound have the same leading instance dependence up to logarithmic factors.
\end{corollary}

Uniform pivotality gives
$\mathcal L_{\mathrm{piv}}=\mathcal L$ with bounded block constants,
so Theorem~\ref{thm:structural-quantum-lower} yields the stated lower
bound. Theorem~\ref{thm: query-complexity} yields the upper bound.

\paragraph{Exact identification, best arm identification, and general trees.}

Although QMCTS is presented for $\epsilon \in(0,1]$, the
formulation and lower-bound analysis also cover exact identification
by setting $\epsilon=0$, provided $\Delta_\star>0$. For a depth-one
MAX tree, this reduces to fixed-confidence exact BAI, while $\epsilon \in(0,1]$ reduces to fixed-confidence
$\epsilon$-BAI. Thus, BAI is a special case of the
tree-search problem studied here.

\section{Classical MCTS Benchmark}
\label{sec:CMCTS}

\begin{theorem}
\label{thm:cmcts-correctness}
With probability at least $1-\delta$, CMCTS stops after finitely many
iterations and returns an $\epsilon$-optimal root move for $s_0$, i.e.,
\[
    \mathbb P\left(
        V_{s_0}-V_{\widehat{s}^{\mathrm C}}\le\epsilon
    \right)
    \ge 1-\delta.
\]
\end{theorem}

For each leaf $\ell$, let
$X_{\ell,1},X_{\ell,2},\ldots$ be an i.i.d. sequence of
Bernoulli observations with mean $\mu_\ell$. For $n\ge1$, define
\[
    \overline X_{\ell,n}
    :=
    \frac{1}{n}\sum_{k=1}^{n}X_{\ell,k}.
\]
Define the classical good event
\[
    \mathcal G_{\mathrm C}
    :=
    \bigcap_{r\ge1}\bigcap_{\ell\in\mathcal L}
    \left\{
        \left|\overline X_{\ell,n_r}-\mu_\ell\right|
        \le\alpha_r
    \right\}.
\]
By Hoeffding's inequality and the definition of $n_r$,
\[
    \mathbb P\left(
        \left|\overline X_{\ell,n_r}-\mu_\ell\right|>\alpha_r
    \right)
    \le
    2\exp(-2n_r\alpha_r^2)
    \le\eta_r.
\]
Consequently,
\[
    \mathbb P(\mathcal G_{\mathrm C}^{\mathrm c})
    \le
    \sum_{r\ge1}\sum_{\ell\in\mathcal L}\eta_r
    =
    \frac{\delta}{2}\sum_{r\ge1}\frac{1}{r^2}
    <\delta.
\]

Condition on $\mathcal G_{\mathrm C}$. Every active classical leaf
estimate has error at most $\alpha_r$. The proofs of
Lemmas~\ref{lem:value-estimation} and~\ref{lem:safe-threshold} therefore
apply with $\widehat{\mu}^{\mathrm C}_{\ell,r}$ and
$\widehat V^{\mathrm C}_{s,r}$ in place of their quantum counterparts.

Because $\gamma_r=2^{-r}$ and $\epsilon>0$, CMCTS stops no later than
the first iteration $r$ for which $\gamma_r\le\epsilon$. If it stops
with only one root move remaining, that move is optimal. Otherwise,
let $s_r^\star$ be an active optimal root move. Since CMCTS returns
the active root move with the largest propagated estimate,
\begin{align*}
    V_{s_0}-V_{\widehat{s}^{\mathrm C}}
    &\le
    \left|
        V_{s_r^\star}
        -\widehat V^{\mathrm C}_{s_r^\star,r}
    \right|
    +
    \left|
        \widehat V^{\mathrm C}_{\widehat{s}^{\mathrm C},r}
        -V_{\widehat{s}^{\mathrm C}}
    \right| \le
    2\alpha_r
    =
    \gamma_r
    \le
    \epsilon.
\end{align*}
Since $\mathbb P(\mathcal G_{\mathrm C})\ge1-\delta$, the result
follows.

\begin{theorem}
\label{thm:cmcts-query-complexity}
With probability at least $1-\delta$, the total number of leaf-oracle
queries used by CMCTS is
\begin{equation}
    O\left(
        \sum_{\ell\in\mathcal L}
        \frac{1}{d_{\ell,\epsilon}^2}
        \left[
            \log\frac{4L}{\delta}
            +
            2\log\left(
                2+\log_2\frac{1}{d_{\ell,\epsilon}}
            \right)
        \right]
    \right),
\label{eq:cmcts-query-upper-bound}
\end{equation}
where
$
    d_{\ell,\epsilon}
    =
    \Delta_\ell\vee\Delta_\star\vee\epsilon.
$
\end{theorem}

Condition on the classical good event $\mathcal G_{\mathrm C}$, and let
$R_\ell$ be the final iteration in which leaf $\ell$ is queried.
The arguments in the proof of
Theorem~\ref{thm: query-complexity} give
\begin{equation}
    \frac{1}{\gamma_{R_\ell}}
    \le
    \frac{4}{d_{\ell,\epsilon}},
    \qquad
    R_\ell
    \le
    2+\log_2\frac{1}{d_{\ell,\epsilon}}.
\label{eq:cmcts-last-query-bounds}
\end{equation}

Using
$
    \alpha_r=\frac{\gamma_r}{2},
    \eta_r=\frac{\delta}{2Lr^2},
$
we obtain
\begin{align*}
    n_{R_\ell}
    &\le
    1+
    \frac{1}{2\alpha_{R_\ell}^2}
    \log\frac{2}{\eta_{R_\ell}} =
    1+
    \frac{2}{\gamma_{R_\ell}^2}
    \left[
        \log\frac{4L}{\delta}
        +
        2\log R_\ell
    \right].
\end{align*}
Applying~\eqref{eq:cmcts-last-query-bounds} yields
\[
    n_{R_\ell}
    \le
    1+
    \frac{32}{d_{\ell,\epsilon}^2}
    \left[
        \log\frac{4L}{\delta}
        +
        2\log\left(
            2+\log_2\frac{1}{d_{\ell,\epsilon}}
        \right)
    \right].
\]
Summing over all leaves proves
(\ref{eq:cmcts-query-upper-bound}). Finally,
$
    \mathbb P(\mathcal G_{\mathrm C})\ge1-\delta,
$
so the bound holds with probability at least $1-\delta$.

\section{Quantum Algorithm Design}
\subsection{Hybrid Algorithm}
In this subsection, we provide implementation details for the Hybrid MCTS algorithm.

\begin{algorithm}[htbp]
\scriptsize
\caption{Hybrid Monte Carlo Tree Search (Hybrid MCTS)}
\label{alg:hybrid-mcts}
\begin{algorithmic}[1]

\Require Finite MAX--MIN tree $\mathcal{T}$,
accuracy $\epsilon\in(0,1]$, and confidence
$\delta\in(0,1/2)$

\State $\mathcal{T}_0\gets\mathcal{T}$ and
$L\gets|\mathcal{L}|$
\State $N_{\ell}\gets0$ for every $\ell\in\mathcal{L}$
\State $n_0\gets0$ and $r\gets1$

\Comment{Classical phase}

\While{\textbf{true}}

    \State $\gamma_r\gets2^{-r}$,
    $\alpha_r\gets\gamma_r/2$, and
    $\eta_r\gets\delta/(2Lr^2)$

    \State
    $n_r\gets
    \left\lceil
    \log(2/\eta_r)/(2\alpha_r^2)
    \right\rceil$

    \If{a fresh QMC estimate with parameters
    $\alpha_r$ and $\eta_r$ uses fewer than
    $n_r-n_{r-1}$ queries}
        \State \textbf{break}
    \EndIf

    \ForAll{$\ell\in\mathcal{L}_{r-1}$}

        \State
        $\Delta_{\ell}\gets
        \max\{0,n_r-N_{\ell}\}$

        \State Query leaf $\ell$ for
        $\Delta_{\ell}$ additional observations and retain them

        \State
        $N_{\ell}\gets N_{\ell}+\Delta_{\ell}$

        \State Construct $\widehat{\mu}_{\ell,r}$
        as the empirical mean of all
        $N_{\ell}$ retained observations

    \EndFor

    \State Propagate $\widehat{V}_{s,r}$ upward using
    (\ref{eq:value-propagation})

    \ForAll{internal nodes $s$, in reverse topological order}
        \State Remove each child $c$ satisfying
        $\widehat{g}_r(s,c)>\gamma_r$
    \EndFor

    \State Let $\mathcal{T}_r$ be the remaining connected subtree,
    with active leaf set $\mathcal{L}_r$

    \State
    $\widehat{s}_r
    \in
    \operatorname*{arg\,max}_{s\in\mathcal{C}_r(s_0)}
    \widehat{V}_{s,r}$

    \If{$|\mathcal{C}_r(s_0)|=1$ \textbf{or}
        $\gamma_r\le\epsilon$}
        \State \Return $\widehat{s}_r$
    \EndIf

    \State $r\gets r+1$

\EndWhile

\Comment{Quantum phase}

\While{\textbf{true}}

    \State $\gamma_r\gets2^{-r}$,
    $\alpha_r\gets\gamma_r/2$, and
    $\eta_r\gets\delta/(2Lr^2)$

    \ForAll{$\ell\in\mathcal{L}_{r-1}$}
        \State
        $\widehat{\mu}_{\ell,r}
        \gets\operatorname{QMC}
        (U_\ell,\alpha_r,\eta_r)$
    \EndFor

    \State Propagate $\widehat{V}_{s,r}$ upward (\ref{eq:value-propagation})

    \ForAll{internal nodes $s$, in reverse topological order}
        \State Remove each child $c$ satisfying
        $\widehat{g}_r(s,c)>\gamma_r$
    \EndFor

    \State Let $\mathcal{T}_r$ be the remaining connected subtree,
    with active leaf set $\mathcal{L}_r$

    \State
    $\widehat{s}_r
    \in
    \operatorname*{arg\,max}_{s\in\mathcal{C}_r(s_0)}
    \widehat{V}_{s,r}$

    \If{$|\mathcal{C}_r(s_0)|=1$ \textbf{or}
        $\gamma_r\le\epsilon$}
        \State \Return $\widehat{s}_r$
    \EndIf

    \State $r\gets r+1$

\EndWhile

\end{algorithmic}
\end{algorithm}

\subsection{Query Complexity under Different Leaf Distributions}
\label{sec: distribution}
In this section, we provide additional numerical results to show that our algorithms also apply to more general leaf distributions. 

We consider two additional distributional settings in the synthetic experiment. The first is a bounded three-point distribution supported on $0$, $1/2$ and $1$, with probabilities chosen to match the mean in the Bernoulli setting. The second is a bounded five-point discrete distribution with equally spaced support points chosen to preserve the same mean. All other experimental settings remain unchanged. As shown in Figures~\ref{fig:three-point} and~\ref{fig:five-point}, the results are consistent with those obtained under Bernoulli distributions, which suggests that our methods are robust to the choice of output distribution.

\begin{figure}
    \centering
    \includegraphics[width=1.0\linewidth]{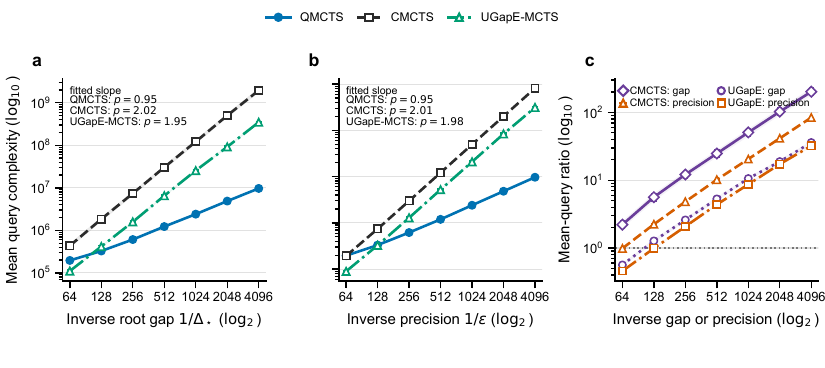}
    \caption{Query Complexity Scaling under A Bounded Three-point
Distribution}
    \label{fig:three-point}
\end{figure}

\begin{figure}
    \centering
    \includegraphics[width=1.0\linewidth]{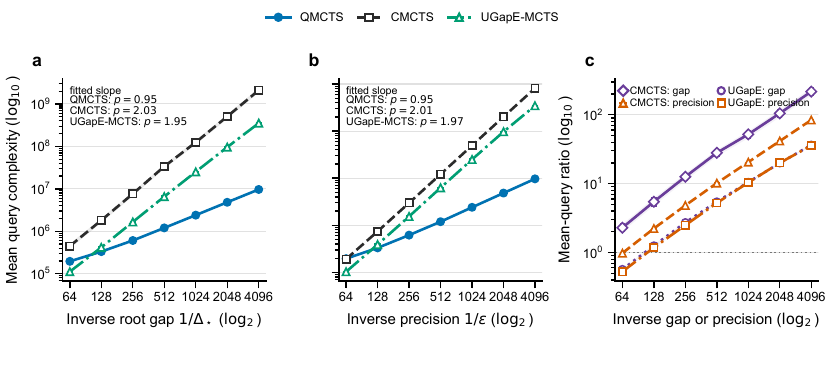}
    \caption{Query Complexity Scaling under A Bounded Five-point Discrete Uniform
Distribution}
    \label{fig:five-point}
\end{figure}




\end{document}